\documentclass[conference, twocolumn,20 pt]{IEEEtran}
\usepackage{array,graphicx,epsfig,amssymb,amsfonts,bbm,amsmath,epstopdf,algorithm,algorithmic,subfigure,cite,tabularx}

\usepackage{latexsym}

\newcolumntype{L}[1]{>{\raggedright\let\newline\\\arraybackslash\hspace{0pt}}m{#1}}
\newcolumntype{C}[1]{>{\centering\let\newline\\\arraybackslash\hspace{0pt}}m{#1}}
\newcolumntype{R}[1]{>{\raggedleft\let\newline\\\arraybackslash\hspace{0pt}}m{#1}}

\newtheorem{theorem}{\textbf{Theorem}}
\newtheorem{lemma}{\textbf{Lemma}}

\newenvironment{proof}[1][Proof]{\begin{trivlist}
\item[\hskip \labelsep {\bfseries #1}]}{\end{trivlist}}

\newcommand{\qed}{\nobreak \ifvmode \relax \else
      \ifdim\lastskip<1.5em \hskip-\lastskip
      \hskip1.5em plus0em minus0.5em \fi \nobreak
      \vrule height0.75em width0.5em depth0.25em\fi}
\begin{document}

\title{A Simple and General Problem and its Optimal Randomized Online Algorithm Design with Competitive Analysis}

\author{Ying ZHANG, IE of CUHK}
\maketitle

\begin{abstract}
The online algorithm design was proposed to handle the caching problem when the future information is unknown \cite{karlin1988competitive}. And currently, it draws more and more attentions from the researchers from the areas of microgrid, where the production of renewables are unpredictable, \cite{lu2013simple},\cite{minghua_sigmetrics}, etc.

In this note, we present a framework of randomized online algorithm design for the \textit{simple and tractable} problem. This framework hopes to provide a tractable design to design a randomized online algorithm, which can be proved to achieve the best competitive ratio by \textit{Yao's Principle} \cite{YaoPrinciple}. %We also apply the approach to design a randomized online algorithm for $\textbf{FS-PAED}$ with the competitive ration of $\frac{e}{e-1+\beta}$, where $\beta = \min \{1,\frac{p_e^{\min}}{p_g}\}$. This ratio is also a lower bound for any randomized online algorithm.

\end{abstract}

\section{A simple but general problem requiring online solution}
In this note, we consider a simple problem, which needs to be solved in the online manner. Suppose its input can be denoted by the parameter $p\in \mathcal{P}$ and its online algorithm can be denoted by $s\in \mathcal{S}$. For example, in the ski rental problem \cite{karlin1988competitive,karlin1994competitive}, $p$ represents how many times the player goes to ski totally, and $s$ represents how many days the player rents the ski before he buys the ski. In our consideration, $p$ and $s$ can be numbers, vectors or matrixes. \footnote{The problem should be \textbf{simple} enough such that we can characterize its input and its online algorithm by a limited number of parameters.} We use the probability distributions of $p$ and $s$ to denote the randomized input and the randomized online algorithm.

Obviously the \textit{optimal offline} cost is uniquely determined by the input $p$, which we denote as $\text{Cost}_{\text{off}}(p)$, while the \textit{online} cost is jointly determined by the input $p$ and the algorithm $s$, which we denote as $\text{Cost}_{\text{on}}(s,p)$.

The ratio of the online cost and offline cost $R(s,p)=\frac{\text{Cost}_{\text{on}}(s,p)}{\text{Cost}_{\text{off}}}$ evaluates how well the online algorithm $s$ performs on the input $p$: a smaller $R(s,p)$ means better $s$, and $R(s,p)\geq 1$. We assume we can obtain a closed form of $R(s,p)$. \footnote{Again, since the problem is so simple}

\section{A Lower Bound for the Competitive Ratio by \textit{Yao's Principle}}

For a given randomized online algorithm $\textit{A}_r$, we can obtain its competitive ratio by $\text{CR}(\textit{A}_r) = \max_{\textit{input}}\frac{\text{Cost}_{\text{on}}}{\text{Cost}_{\text{off}}}$. To show that this randomized online algorithm is the best in terms of competitive ratio, technically, we need to show that given any other randomized online algorithm, the competitive ratio is larger. This is nontrivial because it is difficult to enumerate all possible randomized online algorithms in the design space, or we can think that it's difficult to enumerate all distributions.

%For one general problem, we denote the input as $p$ and the algorithm as $s$; the feasible regions are $\mathcal{P}$ and $\mathcal{S}$ respectively.

In the following analysis, we denote the randomized online algorithm and the randomized input by two randomized variables $S$ with the distribution $f(s)$ and $P$ with the distribution $g(p)$, which are supported by $\mathcal{S}$ and $\mathcal{P}$ respectively.

For convenience, we define two functions $U_g(s)$ and $V_f(p)$ as follows,
\begin{itemize}
\item[$\bullet$] Given the randomized input $g(p)$, $U_g(s)$ represents the expectation of the ratio when the online algorithm is deterministically $s$, i.e.
    $$U_g(s) = \int \frac{\text{Cost}_{\text{on}}(s,p)}{\text{Cost}_{\text{off}}(p)}g(p)dp .$$
\item[$\bullet$] Given the randomized online algorithm $f(s)$, $V_f(p)$ represents the expectation of the ratio when the input is deterministically $p$, i.e.
    $$V_f(p) = \int \frac{\text{Cost}_{\text{on}}(s,p)}{\text{Cost}_{\text{off}}(p)}f(s)ds .$$
\end{itemize}

\subsection{Yao's Principle}

We have \textit{Yao's Principle} \cite{YaoPrinciple} to obtain a lower bound of the competitive ratio.

\begin{lemma}[Yao's Principle] \label{lemma:Yao}
The competitive ratio of any randomized online algorithm is lower bounded by the ratio of any randomized input and the best deterministic online algorithm, i.e.
$$\max_{g(p)}\min_{s}U_g(s)\leq\min_{f(s)}\max_{p}V_f(p)$$
%$$\min_{\textit{A}_r} \max_{\textit{I}}E\left[\frac{\text{Cost}_{\text{on}}}{\text{Cost}_{\text{off}}}\right]\geq
%\max_{\textit{I}_r} \min_{\textit{A}}E\left[\frac{\text{Cost}_{\text{on}}}{\text{Cost}_{\text{off}}}\right]$$
\end{lemma}

Imagine that we can design an online algorithm with the competitive ratio $\mathbf{R}$, which means that $\mathbf{R}$ is an upper bound for $\mathbf{CR}$ \textbf{and} we can also find a random input, the best deterministic online algorithm for which is also $\mathbf{R}$, which means that $\mathbf{R}$ is a lower bound for \textbf{CR}, we can say that our randomized online algorithm can achieve the smallest competitive ratio,thus optimal in terms of \textbf{CR}.

\subsection{By \textit{min max} inequality}

In fact, the \textit{Yao's Principle} can be viewed as a special case of the more general \textit{min max} inequality \cite{boyd2004convex}\footnote{This inequality is so general that $h(x,y)$ can be any real-valued function and that there is no requirement for the function $h$ (say, whether convex or continuous) and the feasible regions of $x$ and $y$ (say, whether convex or compact).},
$$\max_{y}\min_{x} h(x,y)\leq \min_{x}\max_{y}h(x,y).$$

Please be noted that the equality does not always hold. If $\max_{y}\min_{x} h(x,y)= \min_{x}\max_{y}h(x,y)$, we say that $h(x,y)$ and the feasible regions of $x,y$ satisfy the strong max-min property(or the saddle-point property).

Here we define a function $H(f,g) = \int R(s,p)f(s)g(p)dpds$, where $R(s,p) = \frac{\text{Cost}_{\text{on}}(s,p)}{\text{Cost}_{\text{off}}(p)}$, and the variables $f,g$ are the distributions we define in the previous part. We assume the function $R(s,p)f(s)g(p)$ satisfies the condition of \textit{Fubini} Theorem, meaning we can compute $H(f,g)$ by iterated integrals and we can change the order of the integration. \footnote{This requirement is thought to be general} As a result, we can have
$$H(f,g) = \int U_g(s)f(s) ds = \int V_f(p)g(p) dp.$$

By \textit{min max} inequality, we can have $\max_{g(p)}\min_{f(s)}H(f,g)\leq \min_{f(s)}\max_{g(p)}H(f,g)$. Furthermore, note that

%%\begin{equation*}
%\begin{cases}
%\min_{f(s)}H(f,g) = \min_{f(s)} \int U_g(s)f(s) ds = \min_{s}U_g(s), \\
%\max_{g(p)}H(f,g) = \max_{g(p)} \int V_f(p)g(p) dp = \max_{p}V_f(p)
%\end{cases}
%\end{equation*}
\begin{equation*}
\begin{cases}
\min_{f(s)}H(f,g) = \min_{s}U_g(s),\quad \forall g(p)\\
\max_{g(p)}H(f,g) = \max_{p}V_f(p),\quad \forall f(s)
\end{cases}
\end{equation*}

Then we can establish the inequality in \textit{Yao's Principle}.

\textbf{Remark:} note that, by \textit{Yao's Principle}, we can easily have a lower bound once we choose a randomized input, but we don't how tight the lower bound is. It seems we need to randomly pick a randomized input to obtain a lower bound and randomly pick a randomized algorithm to obtain an upper bound, and we are happy only when we are lucky to make them equal to each other. But this 'trial and error' is not good for at least two reasons,
\begin{itemize}
\item We don't know whether the desired randomized input and algorithm exist or not. Maybe the randomized input and algorithm actually don't exist(the equality in \textit{Yao's Principle} never happens for the specific problem we study), then we spend our whole life on trial and error, unhappily. \textit{Question One:} under what condition is the \textit{Yao's Principle} powerful enough to verify the optimality of the randomized online algorithm?
\item If we just randomly pick the randomized input and algorithms, we need to wait for quite a long time to be happy since the design space is so large. In other words, \textit{Yao's Principle} does not provide a guideline to find the optimal distributions. \textit{Question Two:} given that the equality in \textit{Yao's Principle} holds, how can we find the optimally randomized online algorithm $f^*(s)$ and randomized input $g^*(p)$.
\end{itemize}

The remaining part of this note focus on tackling the above two problems. We firstly give a guideline for searching the randomized input and algorithm with the assumption that they do exist; and then we study the existence problem. 

\subsection{A Sufficient and Necessary Condition}

In this part we try to obtain a sufficient and necessary condition for the best randomized online algorithm, \textbf{under the condition} that there does exist such randomized online algorithm whose optimality can be justified by \textit{Yao's Principle}.

\subsubsection{Two lemmas}

\textbf{Sufficient Condition:}

\begin{lemma}\label{lemma:S_condition}
Suppose there exist a randomized online algorithm $\tilde{f}(s)$ and a randomized input $\tilde{g}(p)$, such that $V_{\tilde{f}}(p) = C_2$ and $U_{\tilde{g}}(s) = C_1$, where $C_1$ and $C_2$ are constants, we can have $C_1 = C_2$. As a result, $\tilde{f}(s)$ is the best randomized online algorithm.
\end{lemma}

\begin{proof}
Let's consider the value $$\mathcal{R} = \int_s\int_pR(s,p)f(s)g(p)dpds.$$ If we calculate $\mathcal{R}$ by firstly doing integral on $P$, we can have
\begin{align*}
\mathcal{R} &= \int_sU_{\tilde{g}}(s)\tilde{f}(s)ds\\
&= C_1 ;
\end{align*}
otherwise, we will have
\begin{align*}
\mathcal{R} &= \int_pV_{\tilde{f}}(p)\tilde{g}(p)ds\\
&= C_2.
\end{align*}

Then $C_1 = C_2$ and the proof is complete.

\end{proof}

\textbf{Necessary Condition:}

\begin{lemma}\label{lemma:N_condition}
Suppose there exist a randomized algorithm $f^*(s)$ and a randomized input $g^*(p)$, such that
$$\min_{s}U_{g^*}(s)=\max_{p}V_{f^*}(p),$$
which means that the optimality of $f^*(s)$ can be justified by Lemma~\ref{lemma:Yao}, then we can have $U_{g^*}(s_1) = U_{g^*}(s_2)$ for any $s_1,s_2\in \{s|f^*(s)>0\}$ and $V_{f^*}(p_1) = V_{f^*}(p_2)$ for any $p_1,p_2\in \{p|g^*(p)>0\}$.
\end{lemma}

\begin{proof}
Let
\begin{align*}
\mathcal{R} &= \int_s\int_pR(s,p)f^*(s)g^*(p)dpds\\
& = \int_sU_{g^*}(s)f^*(s)ds \\
& = \int_pV_{f^*}(p)g^*(p)dp .
\end{align*}
and we can have $\min_{s}U_{g^*}(s)\leq \mathcal{R} \leq \max_{p}V_{f^*}(p)$. Then the following equality automatically holds,
\begin{equation}\label{equ:contradiction}
\min_{s}U_{g^*}(s)= \mathcal{R} = \max_{p}V_{f^*}(p).
\end{equation}

For any $s_1,s_2\in \{s|f^*(s)>0\}$, if $U_{g^*}(s_1) < U_{g^*}(s_2)$, we can have
\begin{align*}
\min_{s}U_{g^*}(s) &< \int_sU_{g^*}(s)f^*(s)ds\\
& = \mathcal{R},
\end{align*}
which is a contradiction with Eq~\ref{equ:contradiction}, then we can have $U_{g^*}(s_1) = U_{g^*}(s_2)$ for any $s_1,s_2\in \{s|f^*(s)>0\}$. The remaining similar result can also be proved in the same way.
\end{proof}

\subsubsection{One guideline}

Once we have Lemma~\ref{lemma:S_condition,lemma:N_condition}, we can immediate come up with the guideline for \textit{Question Two}, as follows,

\begin{itemize}
\item[$\mathbf{f^*(s)}$] \textit{Optimal Randomized algorithm:} Set $V_f(p) = C$, i.e. $\frac{dV_f(p)}{dp} = 0$ and with $\begin{cases}f(s)\geq 0\\\int f(s)ds = 1\end{cases}$, to derive $f^*(s)$. \footnote{The math is relative basic but the calculation can be quite intensive}
\item[$\mathbf{g^*(p)}$] \textit{Optimal Randomized input:} Set $U_g(s) = C$, i.e. $\frac{dU_g(s)}{ds} = 0$ and with $\begin{cases}g(p)\geq 0\\\int g(p)dp = 1\end{cases}$, to derive $g^*(p)$.
\end{itemize}

%\textbf{Remark 1:} We omit the proofs for these two lemmas since they are easy, but we note that these proofs(at least for those the author comes up with) require that the function $R(s,p)f(s)g(p)$ satisfies the condition of \textit{Fubini} Theorem.

\textbf{Remark 1:} We remark that the two lemmas can be used to check whether the equality in \textit{Yao's Principle} holds or not.

\textbf{Remark 2:} Actually, with the assumption that the equality in \textit{Yao's Principle} holds, it seems that if we can find a randomized algorithm to achieve a constant ratio for any input, we can say that the algorithm is optimal\footnote{this seems reasonable for the author, but this assertion is so strong that we don't treat it as a lemma currently, to avoid possible confusion}; but it seems equally difficult to verify that `the equality in \textit{Yao's Principle} holds' without checking the previous two lemmas.

\textbf{Remark 3:} The result in this part already gives us enough motivation, in the process of designing a randomized online algorithm, to find $f^*(s)$ to make $V_{f^*}(p)$ being constant for any input $p$, and also $g^*(p)$. However, it does not guarantee that we could find such distributions. Again, note that the analysis in this subsection is made \textit{under the condition} that there does exist such randomized online algorithm whose optimality can be justified by \textit{Yao's Principle} for the given problem. In other words, \textit{Question One} still has no answer.

%currently we are still not sure whether the lower bound by \textit{Yao's Principle} can be achieved or not, or the inequality in $\max_{g(p)}\min_{s}U_g(s)\leq\min_{f(s)}\max_{p}V_f(p)$ is always strict inequality or not. 

\section{A Sufficient and Necessary Condition for a Tight Lower Bound}

In this section, we want to explore under which condition the lower bound by \textit{Yao's Principle} is tight. As explained above, the lower bound being tight is equivalent to that the strong min max property holds for the inequality
$$\max_{g(p)}\min_{f(s)}H(f,g)\leq \min_{f(s)}\max_{g(p)}H(f,g).$$

More specifically, let us firstly define a saddle point for the function $H(f,g)$ as
\begin{equation*}
\begin{cases}
f^* = \arg\min_{f(s)}H(f,g^*)\\
g^* = \arg\max_{g(p)}H(f^*,g),
\end{cases}
\end{equation*}
and we further have Lemma~\ref{lemma:saddle_point}.

\begin{lemma}\label{lemma:saddle_point}
The lower bound by \textit{Yao's Principle} is tight if and only if there exists a saddle point $(f^*,g^*)$ for the function $H(f,g)$.
\end{lemma}

With this lemma, it remains to determine under which condition the function $H(f,g)$ has a saddle point. But the existence of saddle point can be equally difficult to check.

\subsection{On the Existence of Saddle Point}

\subsubsection{Some Mathematical Theorems}
We review the classic theorems for the existence of a saddle point as follows.

\begin{theorem}[Kneser Theorem]
Let $X$ be a nonempty convex subsect in a \textit{Hausdorff topological} vector space $E$ and $Y$ a nonempty compact and convex subset of a \textit{Hausdorff topological} vector space $F$. Let $f$ be a real valued function defined on $X\times Y$. If $(1)$ the function $x\rightarrow f(x,y)$ is concave on $X$, $(2)$ the function $y\rightarrow f(x,y)$ is lower semicontinuous and convex on $Y$, then
$$\min_{y\in Y}\sup_{x\in X}f(x,y) = \sup_{x\in X}\min_{y\in Y}f(x,y).$$
\end{theorem}

\begin{theorem}[Von Neumann Theorem]
Let $X$ and $Y$ be nonempty compact and convex subsets in a \textit{Hausdorff} locally convex vector spaces $E$ and $F$ respectively and $f$ a real valued function defined on $X\times Y$. Suppose $(1)$ the function $x\rightarrow f(x,y)$ is lower semicontinuous and quasiconvex on $X$, $(2)$ the function $y\rightarrow f(x,y)$ is upper semicontinuous and quasiconcave on $Y$. Then, $f$ has a saddle point.
\end{theorem}

A theorem for the more general cases,
\begin{theorem}[General Theorem]
Let $M$ and $N$ be any spaces, $f$ a function on $M\times N$ that is concave-convex like. If for any $c< \inf \sup f$ there exists a \textbf{finite} subset $X\subset M$ such that for any $\nu \in N$ there is an $x\in X$ with $f(x,\nu)>c$, then $\sup\inf f = \inf\sup f$
\end{theorem}

\subsubsection{Results with Compact Feasible Regions $\mathcal{S}$ and $\mathcal{P}$}

Let us firstly make another assumption that the feasible regions for the deterministic online algorithm and input are compact (bounded and closed). For example, $\mathcal{S}$ and $\mathcal{P}$ are compact subspaces of the Euclidian space (recall that $s$ and $p$ can be vectors or matrix). We provide the following well-established theorem to show the existence of the saddle point.

\begin{theorem}[Glicksberg's theorem]
If $A$ and $B$ are compact sets, and $K$ is an upper semicontinuous or lower semi-continuous function on $A\times B$, then
$$\sup_f\inf_g\int\int K df dg = \inf_g\sup_f\int\int K df dg,$$
where $f$ and $g$ run over Borel probability measures on $A$ and $B$.
\end{theorem}

In \textit{Glicksberg's Theorem}, even though we say $A$ and $B$ are subspaces of Euclidian space, the variables, $f$ and $g$, of the function $K$ do not necessarily lies in the Euclidian Space, just thinking about the probability density distribution of a continuous random variable.

Moreover, in my mind, this theorem can be viewed as a generalization of the \textit{Nash Equilibrium theorem} and a special case of the \textit{Debreu- Glicksberg-Fan Theorem}.

\subsection{\textbf{Remark}}

As we can see, it is not easy for the strong min max inequality to hold. So we are not so confident that the optimality of the randomized online algorithm can always be proved by \textit{Yao's Principle} (suppose the \textit{convexity}, \textit{continuity}, \textit{compactness} conditions are not satisfied).

\section{Generalization}
In this part, we try to generalize the above result to the more complex scenarios, in which the algorithm is so complicated that it can not simply represented by single or several variables.

\iffalse
Two preliminary math concepts are presented as follows,

\begin{enumerate}
\item \textit{norm}: norm is a mapping $\|\cdot\|$ from a vector space $\Omega$ to $\mathbb{R}^+$.
\item \textit{Banach Space}: Banach space is a complete normed vector space.
\end{enumerate}
\fi
To make our life easier (easy to use the well established results, especially \textit{Glicksberg's Theorem}), we make two assumptions as follows.

\begin{enumerate}[1]
\item The input belongs to a Banach Space. For a input vector $u$ indicating a demand sequence, its norm is defined as the optimal offline cost to satisfy the demand, i.e., $\text{norm}(u) = \text{Cost}_{\text{off}}(u)$.
\item The online algorithm also belongs to a Banach Space. We represent one online algorithm as a function $f$ from the space of input to $\mathbb{R}^+$, and the value of the function is defined as the online cost given the input $u$, i.e., $f(u) = \text{Cost}_{\text{off}}(u).$ The norm of the function is defined as $\text{norm}(f) = \sup_u\frac{f(u)}{\text{norm}(u)}$
\end{enumerate}

\section{Not the End}
If the above definition is valid (the definition of space and norm need to verify.), the optimal online algorithm can be derived under the framework of this note and its optimality can also be prove if the condition of \textit{Glicksberg's Theorem} is satisfied.

Then we make a conjecture as follows,
\begin{quote}
\textbf{Conjecture: } There exist some problems, the optimality of whose online algorithm cannot be proved by \textit{Yao's Principle}.
\end{quote} 
\bibliographystyle{abbrv}
\bibliography{ref}

\end{document}